\documentclass[final,3p,times]{elsarticle} 

\usepackage{latexsym} 
\usepackage{amsmath} 
\usepackage{epsfig}
\usepackage{amssymb}
\usepackage{enumerate}
\usepackage{expdlist}
\usepackage[parfill]{parskip}
\usepackage{color}

\newtheorem{theorem}{Theorem}[section]
\newtheorem{lemma}[theorem]{Lemma}

\newtheorem{corollary}[theorem]{Corollary}

\newtheorem{observation}[theorem]{Observation}
\newproof{proof}{Proof}

\newcommand{\cA}{{\mathcal A}}

\newcommand{\cC}{{\mathcal C}}

\newcommand{\cI}{{\mathcal I}}

\newcommand{\cL}{{\mathcal L}}
\newcommand{\cN}{{\mathcal N}}
\newcommand{\cO}{{\mathcal O}}
\newcommand{\cP}{{\mathcal P}}

\newcommand{\cS}{{\mathcal S}}
\newcommand{\cT}{{\mathcal T}}

\newcommand{\mrca}{{\rm mrca}}

\usepackage{amssymb}

\journal{XXXX}

\begin{document}

\begin{frontmatter}



\title{On the existence of a cherry-picking sequence}


\author[label1]{Janosch D\"ocker}
\ead{janosch.doecker@uni-tuebingen.de}

\author[label2]{Simone Linz}
\ead{s.linz@auckland.ac.nz}

\address[label1]{Department of Computer Science, University of T\"ubingen, Germany}
\address[label2]{Department of Computer Science, University of Auckland, New Zealand}

\begin{abstract}
Recently, the minimum number of reticulation events that is required to simultaneously embed a collection $\cP$ of rooted binary phylogenetic trees into a so-called temporal network has been characterized in terms of cherry-picking sequences. Such a sequence is a particular ordering on the leaves of the trees in $\cP$. However, it is well-known that not all collections of phylogenetic trees have a cherry-picking sequence. In this paper, we show that the problem of deciding whether or not $\cP$ has a cherry-picking sequence is NP-complete for when $\cP$ contains at least eight rooted binary phylogenetic trees. Moreover, we use automata theory to show that the problem can be solved in polynomial time if the number of trees in $\cP$ and the number of cherries in each such tree are bounded by a constant.
\end{abstract}


\end{frontmatter}


\section{Introduction}
To represent evolutionary relationships among species, phylogenetic trees have long been a powerful tool. However, as we now not only acknowledge speciation but also non-tree-like processes such as hybridization and lateral gene transfer to be driving forces in the evolution of certain groups of organisms (e.g. bacteria, plants, and fish)~\cite{mallet16,soucy15}, phylogenetic networks become more widely used to represent ancestral histories. A phylogenetic network is a generalization of a rooted phylogenetic tree. More precisely, such a network is a rooted directed acyclic graph whose leaves are labeled~\cite{huson10}.

The following optimization problem, which is biologically relevant and mathematically challenging, motivates much of the theoretical work that has been done in reconstructing phylogenetic networks from phylogenetic trees. Given a collection $\cP$ of rooted binary phylogenetic trees on a set of species such that $\cP$ correctly represents the tree-like evolution of different parts of the species' genomes, what is the smallest number of reticulation events that is required to simultaneously embed the trees in $\cP$ into a phylogenetic network? Here, reticulation events are collectively referring to all non-tree-like events and they are represented by vertices in a phylogenetic network whose in-degree is at least two. Without any structural constraints on a phylogenetic network, it is well-known that $\cP$ can always be embedded into such a network~\cite{baroni05,semple07} and, hence, the optimization problem is well-defined. Moreover, despite the problem being NP-hard~\cite{bordewich07}, even for when $|\cP|=2$, several exact algorithms have been developed that, given two rooted phylogenetic trees, construct a phylogenetic network whose number of reticulation events is minimized over the space of all networks that embed both trees~\cite{albrecht12,chen13,piovesan12,wu10}. 

Motivated by the introduction of temporal networks~\cite{baroni06,moret04}, which are phylogenetic networks that satisfy several time constraints, Humphries et al.~\cite{humphries13,humphries13a} recently investigated the special case of the aforementioned optimization problem for when one is interested in minimizing the number of reticulation events over the smaller space of all temporal networks that embed a given collection of rooted binary phylogenetic trees. More precisely, in the context of their two papers, the authors considered {\it temporal networks} to be phylogenetic networks  that satisfy the following three constraints:
\begin{enumerate}[(1)]
\item speciation events occur successively,
\item reticulation events occur instantaneously, and 
\item each non-leaf vertex has a child whose in-degree is one. 
\end{enumerate}
The second constraint implies that the three species that are involved in a reticulation event, i.e. the new species resulting from this event and its two distinct parents, must coexist in time. Moreover, a phylogenetic network that satisfies the third constraint (but not necessarily the first two constraints) is referred to as a {\it tree-child} network in the literature~\cite{cardona12}. Intuitively, if a phylogenetic network $\cN$ is temporal, then one can assign a time stamp to each of its vertices such that the following holds for each edge $(u,v)$ in $\cN$. If $v$ is a reticulation, then the time stamp assigned to $u$ is the same as the time stamp assigned to $v$. Otherwise, the time stamp assigned to $v$ is strictly greater than that assigned to $u$. Baroni et al.~\cite{baroni06} showed that it can be checked in polynomial time whether or not a given phylogenetic network satisfies the first two constraints. 

Humphries et al.~\cite{humphries13} have established a new characterization to compute the minimum number of reticulation events that is needed to simultaneously embed an arbitrarily large collection $\cP$ of rooted binary phylogenetic trees into a temporal network. This characterization, which is formally defined in Section~\ref{sec:prelim}, is in terms of {\it cherries},  and the existence of a particular type of sequence on the leaves of the trees, called a {\it cherry-picking sequence}. It was shown that such a sequence for $\cP$ exists if and only if the trees in $\cP$ can simultaneously be embedded into a temporal network~\cite[Theorem 1]{humphries13}. Moreover, a cherry-picking sequence for $\cP$ can be exploited further to compute the minimum number of reticulation events that is needed over all temporal networks. Importantly, not every collection $\cP$ is guaranteed to have a solution, i.e. there may be no cherry-picking sequence for $\cP$ and, hence no temporal network that embeds all trees in $\cP$.
It was left as an open problem by Humphries et al.~\cite{humphries13} to analyze the computational complexity of deciding whether or not $\cP$ has a cherry-picking sequence for when $|\cP|=2$.

In this paper, we make progress towards this question and show that it is NP-complete to decide if $\cP$ has a cherry-picking sequence for when $|\cP|\geq 8$. Translated into the language of phylogenetic networks, this result directly implies that it is computationally hard to decide if a collection of at least eight rooted binary phylogenetic trees can simultaneously be embedded into a temporal network. To establish our result, we use a reduction from a variant of the {\sc Intermezzo} problem~\cite{guttmann06}. On a more positive note, we show that deciding if $\cP$ has a cherry-picking sequence can be done in polynomial time if the number of trees and the number of cherries in each such tree are bounded by a constant. To this end, we explore connections between phylogenetic trees and automata theory and show how the problem at hand can be solved by using a deterministic finite automaton.

The remainder of the paper is organized as follows.  The next section contains notation and terminology that is used throughout the paper. Section~\ref{sec:intermezzo} establishes NP-completeness of a variant of the {\sc Intermezzo} problem which is then, in turn, used in Section~\ref{sec:cps} to show that it is NP-complete to decide if $\cP$ has a cherry-picking sequence {for when $|\cP|\geq 8$}. In Section~\ref{sec:cherries}, we show that deciding if $\cP$ has a cherry-picking sequence is polynomial-time solvable if the number of cherries in each tree and the size of $\cP$ are bounded by a constant. We finish the paper with some concluding remarks in Section~\ref{sec:con}.

\section{Preliminaries}\label{sec:prelim}
This section provides notation and terminology that is used in the subsequent sections. Throughout this paper, $X$ denotes a finite set.

\noindent{\bf Phylogenetic trees.}  A {\em rooted binary phylogenetic $X$-tree} $\cT$ is a rooted tree with leaf set $X$ and, apart from the root which has degree two, all interior vertices have degree three. Furthermore, a pair of leaves $\{a,b\}$ of $\cT$ is called a {\it cherry} if $a$ and $b$ are leaves that are adjacent to a common vertex. Note that every rooted binary phylogenetic tree has at least one cherry. We denote by $c_\cT$ the number of cherries in $\cT$. We now turn to a rooted binary phylogenetic tree with exactly one cherry. More precisely, we call $\cT$ a {\it caterpillar} if $|X|=n\geq 2$ and the elements in $X$ can be ordered, say $x_1, x_2, \ldots, x_n$, so that $\{x_1, x_2\}$ is a cherry and, if $p_i$ denotes the parent of $x_i$, then, for all $i\in \{3, 4, \ldots, n\}$, we have $(p_i, p_{i-1})$ as an edge in $\cT$, in which case we denote the  caterpillar by $(x_1, x_2, \ldots, x_n)$. To illustrate, Figure~\ref{fig:caterpillar} shows the caterpillar $(D_1,D_2,\ldots,D_{|A'|})$ with cherry $\{D_1,D_2\}$. Two rooted binary phylogenetic $X$-trees $\cT$ and $\cT'$ are said to be {\it isomorphic} if the identity map on $X$ induces a graph isomorphism on the underlying trees.

\noindent{\bf Subtrees.} Now, let $\cT$ be a rooted binary phylogenetic $X$-tree, and let $X'=\{x_1,x_2,\ldots,x_k\}$ be a subset of $X$. The minimal rooted subtree of $\cT$ that connects all vertices in $X'$ is denoted by $\cT(X')$.
Furthermore, the rooted binary phylogenetic tree obtained from $\cT(X')$ by contracting all non-root
degree-$2$ vertices is the {\it restriction of $\cT$ to $X'$} and is denoted by $\cT|X'$. We also write $\cT[-x_1,x_2,\ldots,x_k]$ or $\cT[-X']$ for short to denote $\cT|(X-X')$. For a set $\cP=\{\cT_1,\cT_2,\ldots,\cT_m\}$ of rooted binary phylogenetic $X$-trees, we write $\cP|X'$ (resp. $\cP[-X']$) when referring to the set  $\{\cT_1|X',\cT_2|X',\ldots,\cT_m|X'\}$ (resp.~$\{\cT_1[-X'],\cT_2[-X'],\ldots,\cT_m[-X']\}$).  Lastly, a rooted binary phylogenetic tree is {\it pendant} in $\cT$ if it can be detached from $\cT$ by deleting a single edge.

\noindent{\bf Cherry-picking sequences.} Let $\cP$ be a set of rooted binary phylogenetic $X$-trees with $|X|=n$. We say that an ordering of the elements in $X$, say $(x_1,x_2,\ldots,x_n)$, is a {\it cherry-picking sequence} for $\cP$ precisely if each $x_i$ with $i\in\{1,2,\ldots,n-1\}$ labels a leaf of a cherry in each tree that is contained in $\cP[-x_1,x_2,\ldots,x_{i-1}]$. Clearly, if $|\cP|=1$, then $\cP$ has a cherry-picking sequence. However, if $|\cP|>1$, then $\cP$ may or may not have a cherry-picking sequence.

We now formally state the decision problem that this paper is centered around.

\noindent{\sc CPS-Existence} \\
{\bf Instance.} A collection $\cP$ of rooted binary phylogenetic $X$-trees. \\
{\bf Question.} Does there exist a cherry-picking sequence for $\cP$?

\noindent The significance of {\sc CPS-Existence} is the problem's equivalence to the question whether or not all trees in $\cP$ can simultaneously be embedded into a rooted phylogenetic network that satisfies the three temporal constraints as alluded to in the introduction.

\noindent{\bf Automata and languages.} Let $\Sigma$ be an alphabet. A {\it language} $\cL$ is a subset of all possible strings (also called {\it words}) whose symbols are in $\Sigma$. More precisely, $\cL$ is a subset of $\Sigma^*$, where the operator $*$ is the Kleene star. A {\it deterministic finite automaton} (or short {\it automaton}) 
is a tuple $\cA = (Q, \Sigma, \delta, q_{\text{ini}}, F)$, where
\begin{enumerate}[(i)]
\item $Q$ is a finite set of states,
\item $\Sigma$ is a finite alphabet,  
\item $\delta \colon Q \times \Sigma \rightarrow Q$ is a transition relation,
\item $q_{\text{ini}}$ is the initial state, and
\item $F \subseteq Q$ are final states.  
\end{enumerate}
A  given automaton $\cA$ {\it accepts} a word $w = a_1a_2\ldots a_{n}$ if and only if $\cA$ is in a final state after having read all symbols from left to right, i.e. $$\delta(\ldots\delta(\delta(q_{\text{ini}},a_1),a_2),\ldots a_{n})\in F.$$ The language $\cL(\cA) \subseteq \Sigma^*$ that is {\it recognized} by $\cA$ is defined as the set of words that $\cA$ accepts. 
For the automata constructed in this paper, we have $|F| = 1$ and $\delta$ being a total function that maps each pair of a state in $Q$ and a symbol in $\Sigma$ to a state in $Q$. For a detailed introduction to automata theory and languages, see the book by Hopcroft and Ullman~\cite{hopcroft79}. 


\section{A variant of the {\sc Intermezzo} problem}\label{sec:intermezzo}

In this section, we establish NP-completeness of a variant of the ordering problem {\sc Intermezzo}. Let $A$ be a finite set, and let $\cO$ be an ordering on the elements in $A$. For two elements $a$ and $b$ in $A$, we write $a<b$ precisely if $a$ precedes $b$ in $\cO$. With this notation in hand, we now formally state {\sc Intermezzo} which was shown to be NP-complete via reduction from {\sc 3-SAT}~\cite[Lemma 1]{guttmann06}.

\noindent{\sc Intermezzo} \\
{\bf Instance.} A finite set $A$, a collection $B$ of pairs from $A$, and a collection $C$ of pairwise-disjoint triples of distinct elements in $A$. \\
{\bf Question.} Does there exist a total linear ordering on the elements in $A$ such that $a_i<a_j$ for each $(a_i,a_j)$ in $B$, and $a_i<a_j<a_k$ or $a_j<a_k<a_i$ for each $(a_i,a_j,a_k)$ in $C$?

\noindent {\bf Example.} Consider the following instance of {\sc Intermezzo} with three pairs and two disjoint triples (when viewed as sets):
\begin{align*}
A &= \{a_1,\, a_2,\, a_3,\, a_4,\, a_5, a_6\}, \\
B &= \{(a_1,\, a_6),\, (a_4,\, a_1),\, (a_4,\, a_3)\}, \\
C &= \{(a_1,\,a_2,\,a_3),\,(a_4,\,a_5,\,a_6)\}.
\end{align*}
A total linear ordering on the elements in $A$ that satisfies all constraints defined by $B$ and $C$ is 
\[
\mathcal{O} = (a_2,\, a_4,\, a_3,\, a_1,\, a_5,\, a_6).
\]

\noindent While each element $a_i\in A$ can appear an unbounded number of times in the input of a given {\sc Intermezzo} instance, this number is bounded from above by $N$ in the following {\sc Intermezzo} variant.

\noindent{\sc $N$-Disjoint-Intermezzo} \\
{\bf Instance.} A finite set $A$, collections $B_1,B_2,\ldots,B_N$ of pairs from $A$, and collections $C_1,C_2,\ldots,C_N$ of triples of distinct elements in $A$ such that, for each $\ell \in\{1,2,\ldots,N\}$, the elements in $B_\ell\cup C_\ell$ are pairwise disjoint. \\
{\bf Question.} Does there exist a total linear ordering on the elements in $A$ such that $$a_i<a_j\textnormal{ for each }(a_i,a_j)\in \bigcup_{1 \leq \ell \leq N} B_\ell,$$ and $$a_i<a_j<a_k \textnormal{ or } a_j<a_k<a_i\textnormal{ for each }(a_i,a_j,a_k)\in \bigcup_{1 \leq \ell \leq N} C_\ell\rm{?}$$\\

Let $I$ be an instance of {\sc $N$-Disjoint-Intermezzo}, and let $\cO$ be an ordering on the elements of $A$ that satisfies the two ordering constraints for each pair and triple in the statement of {\sc $N$-Disjoint-Intermezzo}. We say that $\cO$ is an {\it {\sc $N$-Disjoint-Intermezzo} ordering} for $I$. 

We next show that {\sc $4$-Disjoint-Intermezzo} is NP-complete via reduction from the following restricted version of {\sc 3-SAT}.

\noindent{\sc 2P2N-3-SAT} \\
{\bf Instance.} A set $U$ of variables, and a set $\cC$ of clauses, where each clause is a disjunction of exactly three literals, such that each variable appears negated exactly twice  and unnegated exactly twice in $\cC$. \\
{\bf Question.} Does there exist a truth assignment for $U$ that satisfies each clause in $\cC$?

\noindent Berman et al.~\cite[Theorem 1]{berman03} established NP-completeness for {\sc 2P2N-3-SAT}.

\begin{theorem}\label{t:4-DI}
{\sc 4-Disjoint-Intermezzo} is {\em NP}-complete.
\end{theorem}

\begin{proof}
We show that the construction by Guttmann and Maucher~\cite[Lemma 1]{guttmann06}, that was used to show that {\sc Intermezzo} is NP-complete via reduction from {\sc 3-SAT}, yields an instance of {\sc 4-Disjoint-Intermezzo} if we reduce from {\sc 2P2N-3-SAT}.

Using the same notation as Guttmann and Maucher~\cite[Lemma 1]{guttmann06}, their construction is as follows. Let $I$ be an instance of  {\sc 2P2N-3-SAT} that is given by a set of variables $U = \{u_1,\,\ldots,\,u_n\}$ and a set of clauses
\[
\cC=\{(c_{1,1} \vee c_{1,2} \vee c_{1,3}),\,\ldots,\,(c_{m,1} \vee c_{m,2} \vee c_{m,3})\},
\] 
where each $c_{i,j} \in \{u_1, \bar{u}_1,u_2, \bar{u}_2,\ldots,u_n, \bar{u}_n\}$. Furthermore, for $a,b\in\mathbb{N}$, let $a \oplus b$ denote the number $c \in \{1,\,2,\,3\}$ such that $a + b \equiv c \pmod 3$. We define the following three sets:
\begin{align*}
A &= \{u_{k,l},\,\bar{u}_{k,l} \mid 1 \leq k \leq n \wedge 1 \leq l \leq 3\} \cup {}\\
  &\hphantom{{} = {}}\{c_{i,j}^l \mid 1 \leq i \leq m \wedge 1 \leq j \leq 3 \wedge 1 \leq l \leq 3\},\\[5pt]
B &= \{(u_{k,1},\,\bar{u}_{k,3}),\,(\bar{u}_{k,1},\,u_{k,3}) \mid 1 \leq k \leq n\} \cup {}\\
  &\hphantom{{} = {}}\{(c_{i,j,2},\,c_{i,j}^{1}),\,(c_{i,j}^2,\,c_{i,j,1}) \mid 1 \leq i \leq m \wedge 1 \leq j \leq 3\} \cup {}\\
&\hphantom{{} = {}}\{(c_{i,j \oplus 1}^1,\,c_{i,j}^3) \mid 1 \leq i \leq m \wedge 1 \leq j \leq 3\},\\[5pt]
C &= \{(u_{k,1},\,u_{k,2},\,u_{k,3}),\,(\bar{u}_{k,1},\,\bar{u}_{k,2},\,\bar{u}_{k,3}) \mid 1 \leq k \leq n\} \cup {}\\
  &\hphantom{{} = {}}\{(c_{i,j}^1,\,c_{i,j}^2,\,c_{i,j}^3) \mid 1 \leq i \leq m \wedge 1 \leq j \leq 3\},
\end{align*}
where $c_{i,j,l}$ is an abbreviation of $u_{k,l}$ with $u_k = c_{i,j}$. 
By construction, the elements in $C$ are pairwise-disjoint triples of distinct elements in $A$ and, so, the three sets $A$, $B$, and $C$ form an instance of {\sc Intermezzo}.



Now, we show how the pairs and triples in $B \cup C$ can be partitioned into sets $B_\ell \cup C_\ell$ with $B_\ell \subseteq
B$, $C_\ell \subseteq C$, and $1 \leq \ell \leq 4$ such that the elements in $B_\ell\cup C_\ell$ are pairwise disjoint. Recalling that $C$ is a set of pairwise-disjoint triples, we start by setting $B_1 = \emptyset$ and $C_1 = C$. Furthermore, we set 
\begin{align*}
B_2 &= \{(u_{k,1},\,\bar{u}_{k,3}),\,(\bar{u}_{k,1},\,u_{k,3}) \mid 1 \leq k \leq n\} \cup {}\\
&\hphantom{{} = {}}\{(c_{i,j \oplus 1}^1,\,c_{i,j}^3) \mid 1 \leq i \leq m \wedge 1 \leq j \leq 3\}
\end{align*}
and $C_2 = \emptyset$. By construction, it is easy to check that the pairs in $B_2$ are pairwise disjoint. Lastly, consider the remaining pairs
\[
B \setminus B_2 = \{(c_{i,j,2}, c_{i,j}^1), (c_{i,j}^2, c_{i,j,1}) \mid 1 \leq i \leq m \wedge 1 \leq j \leq 3\}
\]
and observe that the only possibility for two pairs in $B \setminus B_2$ to have a non-empty intersection is to have an element $c_{i,j,l}$ with $l\in\{1,2\}$ in common. Now, since each $c_{i,j,l}$ is equal to an element in  $$U'=\{u_{k,l}, \bar{u}_{k,l}\mid 1 \leq k \leq n\wedge 1 \leq l \leq 3\},$$ and each element $u_k$ appears exactly twice negated and twice unnegated in $\cC$, it follows that there is a partition of $B\setminus B_2$ into $B_3$ and $B_4$ so that all pairs in the resulting two sets are pairwise disjoint.
Setting $C_3 = C_4 = \emptyset$ completes the construction of an instance of {\sc 4-Disjoint-Intermezzo}. Noting that it is straightforward to compute the  partition $$B \cup C = \bigcup_{1 \leq \ell \leq 4} \left ( B_\ell \cup C_\ell \right )$$ in polynomial time and that we did not modify the construction described by Guttmann and Maucher~\cite[Lemma 1]{guttmann06} itself, it follows from the same proof that $I$ has a satisfying truth assignment if and only if $\bigcup_{1 \leq \ell \leq 4} \left ( B_\ell \cup C_\ell \right )$ has a {\sc 4-Disjoint-Intermezzo} ordering. \qed
\end{proof}


\noindent {\bf Remark.} By the construction of an instance of {\sc 4-Disjoint-Intermezzo} in the proof of Theorem~\ref{t:4-DI}, we note that no pair or triple occurs twice and that, for each $\ell\in\{1,2,3,4\}$, we have $B_\ell\cup C_\ell\ne\emptyset$. We will freely use these facts throughout the remainder of the paper.

\section{Hardness of {\sc CPS-Existence}}\label{sec:cps}

In this section, we show that the decision problem {\sc CPS-Existence} is NP-complete for any collection of rooted binary phylogenetic trees on the same leaf set that consists of a constant number $m$ of trees with $m\geq 8$. To establish the result, we use a reduction from {\sc 4-Disjoint-Intermezzo}.

\begin{figure}
\centering
\includegraphics[width=.35\textwidth]{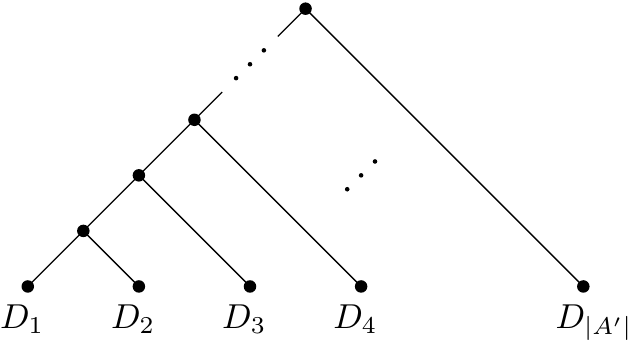}
\caption{A  caterpillar on $|A'|$ leaves and with cherry $\{D_1,D_2\}$.}
\label{fig:caterpillar}
\end{figure}

Let $I$ be an instance of {\sc 4-Disjoint-Intermezzo}. Using the same notation as in the definition of {\sc $N$-Disjoint-Intermezzo}, let $$A'=A\cup\left\{c_r^1,c_r^2,c_r^3,c_r^4\mid c_r\in\bigcup_{1 \leq \ell \leq 4} C_\ell\right\},$$ and let $D=\{d_1,d_2,\ldots,d_{|A'|}\}$. 
For each $\ell\in\{1,2,3,4\}$, we next construct two rooted binary phylogenetic trees. Let $A_\ell$ be the subset of $A'$ that precisely contains each element of $A'$ that is neither contained in an element of $B_\ell$ nor contained in an element of $$C_\ell\cup \{c_r^1,c_r^2,\ldots,c_r^4 \mid c_r\in C_\ell\}.$$
Furthermore, let $\cS_\ell$ and $\cS_\ell'$ both be the caterpillar shown in Figure~\ref{fig:caterpillar}.
Setting $q=1$, let $\cT_\ell$ and $\cT_\ell'$  be the two rooted binary phylogenetic trees obtained from $\cS_\ell$ and $\cS_\ell'$ that result from the following four-step process.
\begin{enumerate}[(i)]
\item For each $(a_i,a_j)\in B_\ell$ in turn, replace the leaf $D_q$ in $\cS_\ell$ (resp. $\cS_\ell'$) with the 3-taxon tree on the top left (resp. bottom left) in Figure~\ref{fig:gadgets} and increment $q$ by one.
\item For each $c_r\in C_\ell$ with $c_r=(a_i,a_j,a_k)$ in turn, replace the leaf $D_q$ in $\cS_\ell$ (resp. $\cS_\ell'$) with the 8-taxon tree on the top right (resp. bottom right) in Figure~\ref{fig:gadgets} and increment $q$ by one.
\item For each $a_i\in A_\ell$ in turn, replace the leaf $D_q$ in $\cS_\ell$ and $\cS_\ell'$ with the cherry $\{a_i,d_q\}$ and increment $q$ by one.
\item For each element in $\{q,q+1,\ldots,|A'|\}$, replace the leaf label $D_q$ in $\cS_\ell$ and $\cS_\ell'$ with $d_q$. 
\end{enumerate}

\begin{figure}
\centering
\includegraphics[width=.6\textwidth]{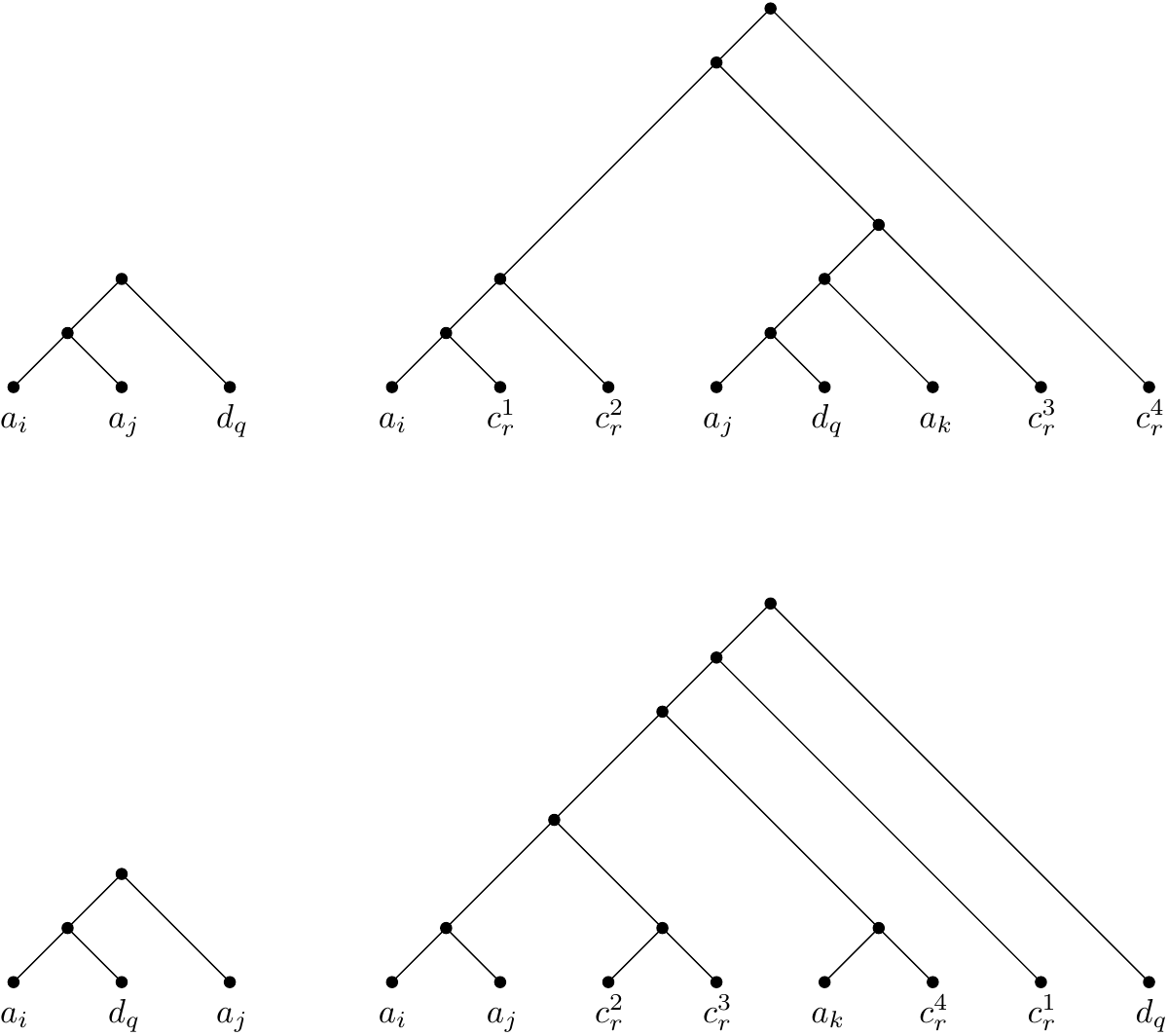}
\caption{Gadgets for a pair $(a_i,a_j)$ (left) and gadgets for a triple $(a_i,a_j,a_k)$ (right) that are used in the reduction from {\sc 4-Disjoint-Intermezzo} to {\sc CPS-Existence}.}
\label{fig:gadgets}
\end{figure}

\noindent We call $\cP_I=\{\cT_\ell,\cT_\ell' \mid 1 \leq \ell \leq 4\}$ the set  of {\it intermezzo trees} associated with $I$. The next observation is an immediate consequence from the above construction and the fact that, for each $1 \leq \ell \leq 4$, the elements in $B_\ell$ and $C_\ell$ are pairwise disjoint.

\begin{observation}
For an instance $I$ of {\sc 4-Disjoint-Intermezzo}, the set of intermezzo trees associated with $I$ consists of eight pairwise non-isomorphic rooted binary phylogenetic trees whose set of leaves is $A'\cup D$.
\end{observation} 


We now establish the main result of this section.

\begin{theorem}\label{lem:CPS-Existence}
Let $\cP=\{\cT_1,\cT_2,\ldots,\cT_m\}$ be a collection of rooted binary phylogenetic $X$-trees. {\sc CPS-Existence} is {\rm NP}-complete for $m=8$.
\end{theorem}

\begin{proof}
Clearly,  {\sc CPS-Existence} for $m=8$ is in NP because, given an ordering $\cO$ on the elements in $X$, we can decide in polynomial time if $\cO$ is a cherry-picking sequence for $\cP$. Let $I$ be an instance of {\sc 4-Disjoint-Intermezzo}, and let $\cP_I=\{\cT_\ell,\cT_\ell'\mid 1 \leq \ell \leq 4\}$ be the set of eight intermezzo trees that are associated with $I$. Note that each  tree in $\cP_I$ can be constructed in polynomial time and has a size that is polynomial in $|A|$. The remainder of the proof essentially consists of establishing the following claim.


\noindent{\bf Claim.} $I$ is a `yes'-instance of {\sc 4-Disjoint-Intermezzo} if and only if $\cP_I$ has a cherry-picking sequence.

First, suppose that $\cP_I$ has a cherry-picking sequence. Let $\cO$ be a cherry-picking sequence for $\cP_I$, and let $\cO'$ be the subsequence of $\cO$ of length $|A|$ that contains each element in $A$. We next show that $\cO'$ is a {\sc $4$-Disjoint-Intermezzo} ordering for $I$. Let $(a_i,a_j)$ be an element of some $B_\ell$ with $1 \leq \ell \leq 4$, and let $d_q$, with $q\in\{1,2,\ldots,|A'|\}$, be the unique leaf label of $\cT_\ell$ and $\cT_\ell'$ such that $\{a_i,a_j,d_q\}$ is the leaf set of a pendant subtree of $\cT_\ell$ and $\cT_\ell'$. By construction of $\cT_\ell$ and $\cT_\ell'$, it is easily seen that $d_q$ exists and $a_i<a_j$ in $\cO$.
Hence, $a_i<a_j$  in $\cO'$. Turning to the triples, let $c_r=(a_i,a_j,a_k)$ be an element of some $C_\ell$ with $1 \leq \ell \leq 4$, and let $d_q$, with $q\in\{1,2,\ldots,|A'|\}$, be the unique leaf label of $\cT_\ell$ and $\cT_\ell'$ such that $\{a_i,a_j,a_k,c_r^1,c_r^2,c_r^3,c_r^4,d_q\}$ is the leaf set of a pendant subtree of $\cT_\ell$ and $\cT_\ell'$. Again, by construction, $d_q$ exists. Let $\cS_\ell=\cT_\ell|\{a_i,a_j,a_k,c_r^1,c_r^2,c_r^3,c_r^4,d_q\}$ and, similarly, let $\cS_\ell'=\cT_\ell'|\{a_i,a_j,a_k,c_r^1,c_r^2,c_r^3,c_r^4,d_q\}$. It is straightforward to check that each cherry-picking sequence for $\cS_\ell$ and $\cS_\ell'$  satisfies either
$$a_i<a_j<a_k,\textnormal{ or }a_j<a_k<a_i.$$ 
Hence, as $\cS_\ell$ and $\cS_\ell'$ are pendant in $\cT_\ell$ and $\cT_\ell'$, respectively, we have $a_i<a_j<a_k$, or $a_j<a_k<a_i$ in $\cO$ and, consequently, in $\cO'$. Since the above argument holds for  each pair and each triple, it  follows that $\cO'$ is a {\sc 4-Disjoint-Intermezzo} ordering for $I$ and, so, $I$ is a `yes'-instance.

Conversely, suppose that $I$ is a `yes'-instance of {\sc 4-Disjoint-Intermezzo}. Let $\cO'$ be a {\sc 4-Disjoint-Intermezzo} ordering on the elements of $A$.  To ease reading, let $$C=\bigcup_{1 \leq \ell \leq 4}C_\ell.$$  Modify $\cO'$ as follows to obtain an ordering $\cO$.
\begin{enumerate}[(1)]
\item Concatenate $\cO'$ with the sequence $(d_1,d_2,\ldots,d_{|A'|})$.
\item For each $c_r=(a_i,a_j,a_k)$ in $C$, do one of the following two depending on the order of $a_i$, $a_j$, and $a_k$ in $\cO'$. If $a_i<a_j<a_k$ in $\cO'$, then replace $a_i$ with $a_i,c_r^2$ and replace $a_k$ with $a_k,c_r^3,c_r^1,c_r^4$. Otherwise, if $a_j<a_k<a_i$, replace $a_k$ with $a_k,c_r^3$ and replace $a_i$ with $a_i,c_r^2,c_r^1,c_r^4$.
\end{enumerate}
Since $\cO'$ is a {\sc 4-Disjoint-Intermezzo} ordering  with $a_i<a_j<a_k$ or $a_j<a_k<a_i$ for each $(a_i,a_j,a_k)\in C$, it follows from the construction of $\cO$ from $\cO'$ that $\cO$ is an ordering on the elements in $A'\cup D$. It remains to show that $\cO$ is a cherry-picking sequence for $\cP_I$. First, consider a  pendant subtree with leaf set $\{a_i,a_j,d_q\}$ in $\cT_\ell$ and $\cT_\ell'$ for some $1 \leq \ell \leq 4$. By construction, $(a_i,a_j)$ is a pair in $B_\ell$ and, so, we have $a_i<a_j$ in $\cO'$ and $a_i<a_j<d_q$ in $\cO$. Second, consider a  pendant subtree with leaf set  $\{a_i,a_j,a_k,c_r^1,c_r^2,c_r^3,c_r^4,d_q\}$ in $\cT_\ell$ and $\cT_\ell'$ for some $1 \leq \ell \leq 4$. By construction, $(a_i,a_j,a_k)$ is a triple in $C_\ell$ and, so, we have either $a_i<a_j<a_k$ in $\cO'$ and $$a_i<c_r^2<a_j<a_k<c_r^3<c_r^1<c_r^4<d_q$$ in $\cO$, or $a_j<a_k<a_i$ in $\cO'$ and $$a_j<a_k<c_r^3<a_i<c_r^2<c_r^1<c_r^4<d_q$$ in $\cO$. Third, consider a  pendant subtree with leaf set $\{a_i,d_q\}$ in $\cT_\ell$ and $\cT_\ell'$ for some $1 \leq \ell \leq 4$. By construction, we have $a_i<d_q$ in $\cO$. Fourth, if $(a_i,a_j,a_k)\in C$, then, as $I$ has a  {\sc 4-Disjoint-Intermezzo} ordering, there does not exist a pair $(a_k,a_j)$ in $B_\ell$ for some $1 \leq \ell \leq 4$. Lastly, observe that $(d_1,d_2,\ldots,d_{|A'|})$ is a suffix of $\cO$ and that, for any two trees, say $\cS$ and $\cS'$ in $\cP_I$, we have that $\cS|D$ and $\cS'|D$ are isomorphic. Since $\cO'$ is a {\sc 4-Disjoint-Intermezzo} ordering,  it is now straightforward to check that $\cO$ is a cherry-picking sequence of $\cP_I$. This establishes the proof of the claim and, thereby, the theorem.\qed
\end{proof}

The next corollary shows that {\sc CPS-Existence} is not only NP-complete for a collection of eight rooted binary phylogenetic trees on the same leaf set, but for any such collection with a fixed number $m$ of trees with $m\geq 8$.

\begin{corollary}
Let $\cP=\{\cT_1,\cT_2,\ldots,\cT_m\}$ be a collection of rooted binary phylogenetic $X$-trees. {\sc CPS-Existence} is {\rm NP}-complete for any fixed $m$ with $m\geq 8$.
\end{corollary}

\begin{proof}
Clearly,  {\sc CPS-Existence} for $m= t+8$ with $t \geq 0$ is in NP. To establish the corollary, we show how one can modify the reduction that is described prior to Theorem~\ref{lem:CPS-Existence} to obtain a set of $t+8$ rooted binary phylogenetic trees from an instance of {\sc 4-Disjoint-Intermezzo}. 

Let $I$ be an instance of {\sc 4-Disjoint-Intermezzo}. Throughout the remainder of the proof, we assume that there exists an $1 \leq \ell \leq 4$ such that $|B_\ell \cup C_\ell| > t$. Otherwise, since $t = m - 8$ and $m$ is fixed, it follows that $I$ has a constant  number $c$ of pairs and triples with $c\leq 4(m-8)$ and is solvable in polynomial time. 

Now, let $B_i$ and $C_i$ with  $i \in \{1,\,2,\,3,\,4\}$ be a collection of pairs and triples, respectively, such that $|B_i \cup C_i| > t$. Theorem~\ref{lem:CPS-Existence} establishes the result for when $t=0$. We may therefore assume that $t>0$ and consider two cases. First, suppose that $t$ is even. Replace $B_i$ and $C_i$ in $I$ with a partition of  $B_i \cup C_i$ into $\frac{t}{2} + 1$ sets. Each of the resulting new sets can be split naturally into a collection of pairs and a collection of triples of which at most one is empty. This results in $$\left (\frac t 2+1\right )+(4-1)=\frac t 2 +4$$ collections of pairs and triples, respectively. Now, for each $B_\ell$ and $C_\ell$ with $\ell\in\{1,2,\ldots,\frac t 2+4\}$, construct two rooted binary phylogenetic trees as described in the definition of the set of intermezzo trees associated with $I$.  This yields $$2\left (\frac {t} 2 +4\right )=t+8=m$$ pairwise non-isomorphic trees. Second, suppose that $t$ is odd. Replace $B_i$ and $C_i$ in $I$ with a partition of  $B_i \cup C_i$ into $\frac{t-1}{2} + 1$ sets. 
Additionally, add $B_*=\emptyset$ and $C_*=\emptyset$. Analogous to the first case, this results in $$\left (\frac {t-1} 2+1+1\right )+(4-1)=\frac {t-1} 2 +5$$
collections of pairs and triples, respectively. Again, for each $B_\ell$ and $C_\ell$ with $\ell\in\{1,2,\ldots,\frac {t-1} 2+5\}$ construct two rooted binary phylogenetic trees as described in the definition of the set of intermezzo trees associated with $I$. Noting that the two trees for $B_*$ and $C_*$ are isomorphic, it follows that the construction yields $$2\left (\frac {t-1} 2 +5\right )-1=\frac{2t-2}2+10-1=t+8=m$$ pairwise non-isomorphic trees. Since the proof of Theorem~\ref{lem:CPS-Existence} generalizes to a set of $m$ intermezzo trees, the corollary now follows for both cases.\qed
\end{proof}

\section{Bounding the number of cherries}\label{sec:cherries}

The main result of this section is the following theorem.


\begin{theorem}\label{t:cherry}
Let $\cP=\{\cT_1,\cT_2,\ldots,\cT_m\}$ be a collection of rooted binary phylogenetic $X$-trees. Let $c$ be the maximum element in $\{c_{\cT_1},c_{\cT_2},\ldots,c_{\cT_m}\}$. Then solving {\sc CPS-Existence} for $\cP$ takes time $$O\left (|X|^{m(4c-2)+1} + \sum_{i=1}^m f_i(|X|, c_{\cT_i})\right ),$$ where $f_i(|X|, c_{\cT_i}) \in |X|^{O(c_{\cT_i})}$. In particular, the running time is polynomial in $|X|$ if $c$ and $m$ are constant.
\end{theorem}

Let $\cT$ be a rooted binary phylogenetic $X$-tree. We denote by $\cC(\cT)$ the recursively defined set of trees that contains $\cT$ and $\emptyset$, and that satisfies the following property.

\noindent {\bf (P)} If a tree $\cT'$ is in $\cC(\cT)$ and $\{a,b\}$ is a cherry in $\cT'$, then $\cT'[-a]$ and $\cT'[-b]$ are also contained in $\cC(\cT)$.

\noindent We refer to $\cC(\cT)$ as the  {\it set of cherry-picked trees} of $\cT$. Intuitively, $\cC(\cT)$ contains each tree that can be obtained from $\cT$ by repeatedly deleting a leaf of a cherry.

To establish Theorem~\ref{t:cherry}, we consider the set $\cC(\cT)$ of cherry-picked trees of $\cT$. First, we develop a new vector representation for each tree in $\cC(\cT)$ and show that the size of $\cC(\cT)$ is at most $(|X|+1)^{O(c_\cT)}$.
 We then construct an automaton whose number of states is $|\cC(\cT)|+1$ and that recognizes whether or not a word that contains each element in $X$ precisely once is a cherry-picking sequence for $\cT$. Lastly, we show how to use a product automaton construction to solve {\sc CPS-Existence}  for a set of rooted binary phylogenetic $X$-trees in time that is polynomial if the number of cherries and the number of trees in $\cP$ is bounded by a constant.

We start with a simple lemma, which shows that deleting a leaf of a cherry never increases the number of cherries. 

\begin{lemma}\label{lem:cherries}
Let $\cT$ be a rooted binary phylogenetic $X$-tree, and let $a$ be an element of a cherry in $\cT$. Then, $$c_{\cT}-c_{\cT[-a]}\in\{0,1\}.$$
\end{lemma}

\begin{proof}
Let $b$ be the unique element in $X$ such that $\{a,b\}$ is a cherry in $\cT$. Observe that each cherry of $\cT$ other than $\{a,b\}$ is also a cherry of $\cT[-a]$. Now, let $p$ be the parent of the parent of $a$ in $\cT$, and let $c$ be the child of $p$ that is not the parent of $a$. If $c$ is a leaf, then it is easily checked that $\{b,c\}$ is a cherry in $\cT[-a]$ and, so $c_{\cT}-c_{\cT[-a]}=0$. On the other hand, if $c$ is not a leaf, then $b$ is not part of a cherry in $\cT[-a]$ and, so, $c_{\cT}-c_{\cT[-a]}=1$. \qed
\end{proof}

We now define a labeled tree that will play an important role throughout the remainder of this section. Let $\cT$ be a rooted binary phylogenetic $X$-tree with cherries $\{\{a_1,b_1\},\{a_2,b_2\},\ldots,\{a_{c_\cT},b_{c_\cT}\}\}$. Obtain a tree $\cT_\cI$ from $\cT$ as follows.

\begin{enumerate}[Step (1).]
\item Set $\cT_\mathcal{I}$ to be $\cT$.
\item Delete all leaves of $\cT_\mathcal{I}$ that are not part of a cherry. 
\item Suppress any resulting degree-2 vertex. 
\item If the root, say $\rho$, has degree one, delete $\rho$.
\item For each cherry $\{a_i, b_i\}$ with $i\in\{1,2,\ldots,c_\cT\}$, label the parent of $a_i$ and $b_i$ with $i$, and delete the two leaves $a_i$ and $b_i$. 
\item Bijectively label the non-leaf vertices of $\cT_\mathcal{I}$ with $c_\cT+1,c_\cT+2,\ldots,2c_\cT-1$.
\end{enumerate}
\noindent We call $\cT_\cI$ the {\it index tree} of $\cT$. By construction, $\cT_\cI$ is a labeled rooted binary tree that is unique up to relabeling the internal vertices. To illustrate, an example of the construction of an index tree is shown in Figure~\ref{fig:index}. The next observation follows immediately from the construction of an index tree.

\begin{observation}
Let $\cT$ be a rooted binary phylogenetic tree, and let $\cT_\cI$ be the index tree associated with $\cT$. The size of $\cT_\cI$ is $O(c_\cT)$. In particular, if the number of cherries in $\cT$ is constant, the size of $\cT_\cI$ is $O(1)$.
\end{observation}

\begin{figure}
\centering
\includegraphics[width=0.8\textwidth]{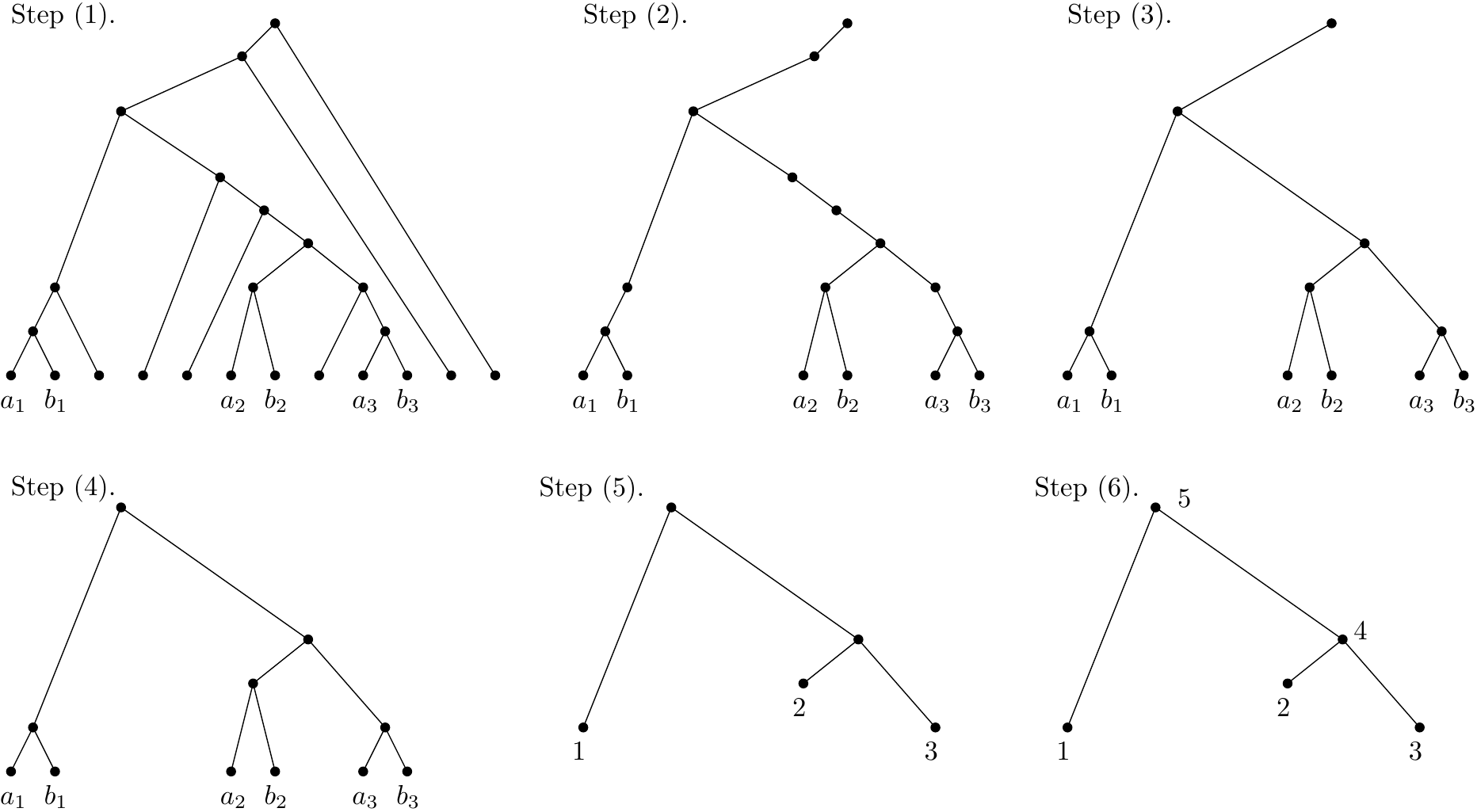}
\caption{An example of the construction of an index tree. Steps (1) to (6) refer to the corresponding steps in the definition of an index tree. For simplicity, in Step (1), we have only indicated the leaf labels of leaves that are part of a cherry.}
\label{fig:index}
\end{figure}

We next define a particular vector relative to a given set. Let $S$ be a finite set, let $\epsilon$ be an element that is not in $S$, and let $n$ be a non-negative integer. We call
\begin{eqnarray}
v=&(&[\xi_1](x_{1}^1,x_{1}^2,\ldots,x_{1}^{q_1},\epsilon),\nonumber\\
&&[\xi_2](x_{2}^1,x_{2}^2,\ldots,x_{2}^{q_2},\epsilon),\nonumber\\
&&\vdots \nonumber\\
&&[\xi_n](x_n^1,x_n^2,\ldots,x_n^{q_n},\epsilon)\hspace{0.3cm})\nonumber
\end{eqnarray}
an {\it $S$-vector} if each element in $S$ appears at most once in $v$, each $\xi_i$ is an element in $S\cup\{\epsilon\}$, and each $x_i^j$ is an element in $S$.
Now consider the following two $S$-vectors:
\begin{eqnarray}
v=&(&[\xi_1](x_{1}^1,x_{1}^2,\ldots,x_{1}^{q_1},\epsilon),\nonumber\\
&&[\xi_2](x_{2}^1,x_{2}^2,\ldots,x_{2}^{q_2},\epsilon),\nonumber\\
&&\vdots \nonumber\\
&&[\xi_n](x_n^1,x_n^2,\ldots,x_n^{q_n},\epsilon)\hspace{0.3cm})\nonumber
\end{eqnarray}
and
\begin{eqnarray}
v'=&(&[\psi_1](y_{1}^1,y_{1}^2,\ldots,y_{1}^{r_1},\epsilon),\nonumber\\
&&[\psi_2](y_{2}^1,y_{2}^2,\ldots,y_{2}^{r_2},\epsilon),\nonumber\\
&&\vdots \nonumber\\
&&[\psi_n](y_n^1,y_n^2,\ldots,y_n^{r_n},\epsilon)\hspace{0.3cm})\nonumber
\end{eqnarray}
We say that $v'$ has the {\it suffix-property} relative to $v$ if, for each $s\in\{1,2,\ldots,n\}$, the vector component $[\psi_s](y_{s}^1,y_{s}^2,\ldots,y_{s}^{r_s},\epsilon)$ is equal to $[\psi_s](\epsilon)$ or satisfies each of the following equations
\[
y_{s}^{r_s} = x_{s}^{q_s}, \; y_{s}^{r_s-1} = x_{s}^{q_s-1}, \; \ldots, \; y_{s}^{1} =  x_{s}^{q_s-r_s+1}.
\] 
Lastly, if $v'$ has the property that $[\psi_i](y_{i}^1,y_{i}^2,\ldots,y_{i}^{r_i},\epsilon)=[\epsilon](\epsilon)$ for each $i\in\{1,2,\ldots,n\}$, we call $v'$ the {\it empty vector}. Note that the empty vector satisfies the suffix-property relative to every $S$-vector.

Building on the definition of an $S$-vector, we now describe a vector representation of a rooted binary phylogenetic tree that can be constructed by using its index tree as a guide. Roughly, the representation associates a caterpillar-type structure to each vertex in the index tree. 
Let $\cT$ be a rooted binary phylogenetic $X$-tree, let $X'\subseteq X$, and let $\epsilon\notin X$. For two vertices $u$ and $v$ in $\cT$, we say that $u$ (resp. $v$) is an {\it ancestor} (resp. {\it descendant}) of $v$ (resp. $u$) if there is a directed path from $u$ to $v$ in $\cT$.  Throughout this section, we regard a vertex $v$ of $\cT$ to be an ancestor and a descendant of itself. 
The {\em most recent common ancestor} of $X'$ is the vertex $v$ in $\cT$ whose set of descendants contains $X'$ and no descendant of $v$, except $v$ itself, has this property. We denote $v$ by $\mrca_{\cT}(X')$. 
Now, let $\{\{a_1,b_1\},\{a_2,b_2\},\ldots,\{a_{c_\cT},b_{c_\cT}\}\}$ be the set of all cherries in $\cT$. First, for each leaf $i\in\{1,2,\ldots,c_\cT\}$ in $\cT_\cI$, let $(a_i,x_{i}^1,x_{i}^2,\ldots,x_{i}^q)$ be the maximal pendant caterpillar in $\cT$ with cherry $\{a_i,b_i\}$. We denote this by $$[\xi_i](x_{i}^1,x_{i}^2,\ldots,x_{i}^q,\epsilon),$$ where $\xi_i=a_i$ and  $x_{i}^1=b_i$.  Second, for each non-leaf vertex labeled $i$ in $\cT_\cI$ with $i\in\{c_\cT+1,c_\cT+2,\ldots,2c_\cT-1\}$, let $v_i$ be the vertex in $\cT$ such that $$v_i={\rm mrca}_\cT(\{a_j,b_j\mid j \textnormal{ is a descendant of } i \textnormal{ in } \cT_\cI\}),$$ and let $\cT_i$ be the rooted binary phylogenetic tree obtained from $\cT$ by replacing the pendant subtree rooted at $v_i$ with a leaf labeled $v_i$. Now, if $v_i$ is a leaf of a cherry in $\cT_i$, let $(v_i,x_{i}^1,x_{i}^2,\ldots,x_{i}^q)$ be the maximal pendant caterpillar in $\cT_i$ with cherry $\{v_i,x_{i}^1\}$. We denote this by $$[\epsilon](x_{i}^1,x_{i}^2,\ldots,x_{i}^q,\epsilon).$$ Otherwise, if $v_i$ is not a leaf of a cherry in $\cT_i$, we denote this by $$[\epsilon](\epsilon).$$ Now, recall that $2c_\cT-1$ is the number of vertices in $\cT_\cI$. Setting $n=2c_\cT-1$, we call 
\begin{eqnarray}
v_\cT=&(&[\xi_1](x_{1}^1,x_{1}^2,\ldots,x_{1}^{q_1},\epsilon),\nonumber\\
&&[\xi_2](x_{2}^1,x_{2}^2,\ldots,x_{2}^{q_2},\epsilon),\nonumber\\
&&\vdots\nonumber\\
&&[\xi_{n}](x_{n}^1,x_{n}^2,\ldots,x_{n}^{q_{n}},\epsilon)\hspace{0.3cm})\nonumber
\end{eqnarray}
the {\it vector representation of $\cT$ relative to $\cT_\cI$}, and note that $$\xi_{c_\cT+1}=\xi_{c_\cT+2}=\cdots=\xi_{2c_\cT-2}=\xi_{n}=\epsilon.$$ An example of a tree and its vector representation is shown in Figure~\ref{fig:vector}.

\begin{figure}
\centering
\includegraphics[width=0.35\textwidth]{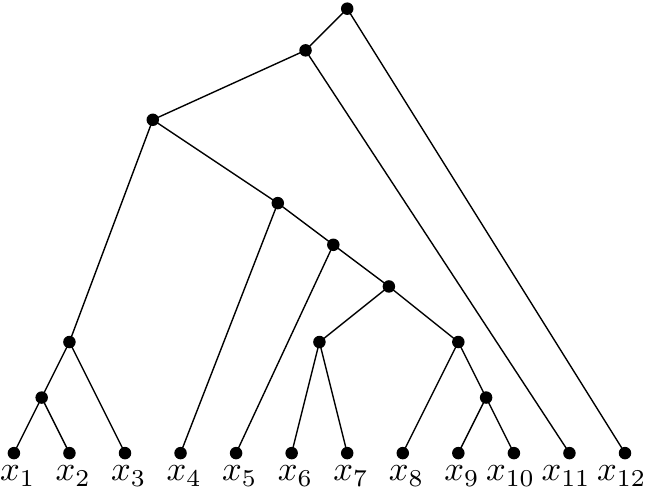}
\caption{A rooted binary phylogenetic tree $\cT$ whose index tree $\cT_\cI$  is shown in Step (6) of Figure~\ref{fig:index}. The vector representation of $\cT$ relative to $\cT_\cI$ is $([x_1](x_2,x_3,\epsilon),[x_6],(x_7,\epsilon),[x_9](x_{10},x_8,\epsilon),[\epsilon](x_5,x_4,\epsilon),[\epsilon](x_{11}, x_{12},\epsilon))$.}
\label{fig:vector}
\end{figure}

Let $\cT$ be a rooted binary phylogenetic $X$-tree with $\epsilon\notin X$, and let $\cT_\cI$ be the index tree of $\cT$. Let $v_\cT$  be the vector representation relative to $\cT_\cI$. Furthermore, let $\cT'$ be an element in $\cC(\cT)$, and let 
\begin{eqnarray}
v_{\cT'}=&(&[\psi_1](y_{1}^1,y_{1}^2,\ldots,y_{1}^{r_1},\epsilon),\nonumber\\
&&[\psi_2](y_{2}^1,y_{2}^2,\ldots,y_{2}^{r_2},\epsilon),\nonumber\\
&&\vdots \nonumber\\
&&[\psi_n](y_n^1,y_n^2,\ldots,y_n^{r_n},\epsilon)\hspace{0.3cm})\nonumber
\end{eqnarray}
be an $X$-vector for $\cT'$. We say that  $v_{\cT'}$  has the {\it cherry-property relative to $v_\cT$} if, for each cherry $\{a,b\}$ in $\cT'$, exactly one of the following conditions holds:
\begin{enumerate}[(i)]
\item There is an index $s \in \{1,\,2,\,\ldots,\,n\}$ such that $\{\psi_s, y_{s}^1\} = \{a,b\}$. 
\item There are two distinct indices $s,t \in \{1,\,2,\,\ldots,\,n\}$ such that $\{\psi_s, \psi_t\} = \{a,b\}$, the two corresponding vector components are $[\psi_s](\epsilon)$ and $[\psi_t](\epsilon)$, respectively, there is a vertex labeled $u$ in $\cT_\cI$ whose two children are labeled $s$ and $t$, and $\psi_u = \epsilon$.
\end{enumerate}

To establish Theorem~\ref{t:cherry}, we next prove three lemmas.

\begin{lemma}\label{l:map}
Let $\cT$ be a rooted binary phylogenetic $X$-tree, and let $v_\cT$ be the vector representation of $\cT$ relative to an index tree of $\cT$. Then each tree $\cT''$ in $\cC(\cT)$ can be mapped to an $X$-vector that satisfies the suffix-property and the cherry-property relative to $v_\cT$. Moreover, the mapping is one-to-one.
\end{lemma}

\begin{proof}
Set $n=2c_\cT-1$.  We  define a mapping $f$ from the elements in $\cC(\cT)$ into the set of all $X$-vectors that satisfy the suffix-property and the cherry-property relative to $v_\cT$. First, we map $\cT$ to $v_\cT$ and note that $v_\cT$ satisfies the suffix-property and the cherry-property relative to $v_\cT$. Second, we map the element $\emptyset$ in $\cC(\cT)$ to the empty vector, say $v_\emptyset$,  with $n$ vector components. Again, $v_\emptyset$ satisfies the suffix-property and the cherry-property relative to $v_\cT$. Now, let $\cT''$ be an element in $\cC(\cT)\setminus\{\cT,\emptyset\}$. Recalling the recursive definition of $\cC(\cT)$, there exists a tree $\cT'$ in $\cC(\cT)$ with cherry $\{a,b\}$ such that $\cT'[-a]$ is isomorphic to $\cT''$.
Suppose that $f$ (defined below) has already mapped $\cT'$  to the $X$-vector
\begin{eqnarray}
v_{\cT'}=&(&[\psi_1](y_{1}^1,y_{1}^2,\ldots,y_{1}^{r_1},\epsilon),\nonumber\\
&&[\psi_2](y_{2}^1,y_{2}^2,\ldots,y_{2}^{r_2},\epsilon),\nonumber\\
&&\vdots \nonumber\\
&&[\psi_n](y_n^1,y_n^2,\ldots,y_n^{r_n},\epsilon)\hspace{0.3cm}),\nonumber
\end{eqnarray}
that satisfies the suffix-property as well as the cherry-property relative to $v_{\cT}$. Then $f$  maps $\cT''$ to a vector that can be obtained from $v_{\cT'}$ in one of the following two cases.
\begin{enumerate}[{\bf (M1)}]
\item If there is an index $s \in \{1,\,2,\,\ldots,\,n\}$ such that $\{\psi_s, y_{s}^1\} = \{a, b\}$, then $f$ maps $\cT''$ to a vector $v_{\cT''}$ that is obtained from $v_{\cT'}$ by replacing the vector component $$[\psi_s](y_{s}^1,y_{s}^2,\ldots,y_{s}^{r_s},\epsilon) \text{ with } [b](y_{s}^2,y_{s}^3\ldots,y_{s}^{r_s},\epsilon).$$ 
\item Otherwise, there are two indices $s,t \in \{1,\,2,\,\ldots,\,n\}$ with $s \neq t$ such that $\{\psi_s, \psi_t\} = \{a, b\}$, where the two corresponding components have the form $[\psi_s](\epsilon)$ and $[\psi_t](\epsilon)$, respectively. Furthermore, by construction, there is a vertex labeled  $u$ in $\cT_\cI$ whose two children are the vertices labeled $s$ and $t$ and  $\psi_u=\epsilon$. Then, $f$ maps $\cT''$ to a vector $v_{\cT''}$ that is obtained from $v_{\cT'}$ by replacing each of the two vector components $$[\psi_s](\epsilon)\text{ and } [\psi_t](\epsilon)\text{ with } [\epsilon](\epsilon),$$ and replacing the vector component $$[\epsilon](y_{u}^1,y_{u}^2,\ldots,y_{u}^{r_u},\epsilon) \text{ with } [b](y_{u}^1,y_{u}^2\ldots,y_{u}^{r_u},\epsilon).$$
\end{enumerate}
For both cases, it is easily checked that $v_{\cT''}$ is an $X$-vector that satisfies the suffix-property relative to $v_\cT$. 

We next show that $v_{\cT''}$ satisfies the cherry-property relative to $v_\cT$. By Lemma \ref{lem:cherries}, we have $c_{\cT'}-c_{\cT''}\in\{0,1\}$. If $c_{\cT'}-1=c_{\cT''}$ then, by construction, each cherry in $\cT''$ is a cherry in $\cT'$. Hence, as $v_{\cT'}$ satisfies the cherry-property relative to $v_\cT$, we have that $v_{\cT''}$ satisfies the cherry-property relative to $v_\cT$. Otherwise, if $c_{\cT'}=c_{\cT''}$, then the cherry $\{a,b\}$ in $\cT'$ is replaced with a new cherry that contains $b$, while all other cherries in $\cT'$ are also cherries in $\cT''$ . First, suppose that $\cT''$ is obtained from $\cT'$ according to mapping (M1). Observe that $r_s\geq 1$. If $r_s\geq 2$, then $\{b,y_s^2\}$ is the new cherry and, thus, $v_{\cT''}$ satisfies the cherry-property relative to $v_\cT$. On the other hand, if $r_s=1$, let $[\psi_t](y_t^1,y_t^2,\ldots,y_t^{r_t},\epsilon)$ be the vector component in $v_{\cT'}$ such that the vertices labeled $s$ and $t$ in $\cT_\cI$ have the same parent. Note that $t$ exists because, otherwise, $s$ is the root of $\cT_\cI$ and so the existence of a cherry in $\cT''$ that is not a cherry in $\cT'$ implies that $r_s\geq 2$; a contradiction. If $r_t\geq 1$, then $\{\psi_t,y_t^1\}$ is a cherry in $\cT'$ and $\cT''$. Thus, the sibling of $b$ in $\cT''$ is not a leaf, thereby contradicting that $b$ is a leaf of a cherry in $\cT''$. Hence, $[\psi_t](y_t^1,y_t^2,\ldots,y_t^{r_t},\epsilon)=[\psi_t](\epsilon)$. Now, as  $[b](\epsilon)$ and $[\psi_t](\epsilon)$ are two vector components of $v_{\cT''}$, it again follows that $v_{\cT''}$ satisfies the cherry-property relative to $v_\cT$. Second, suppose that $\cT''$ is obtained from $\cT'$ according to mapping (M2). Noting that $b$ is an element of the vector component $[\psi_u](y_u^1,y_u^2,\ldots,y_t^{r_u},\epsilon)$ with $\psi_u=b$ in $v_{\cT''}$, the result can be established by using an argument that is similar to the previous case. 

It remains to show that the mapping is one-to-one. Let $\cT'$ and $\cT''$ be two distinct elements in $\cC(\cT)\setminus\{\emptyset\}$. Since each element in $\cC(\cT)\setminus\{\cT, \emptyset\}$ can be obtained from $\cT$ by  repeatedly deleting a leaf of a cherry and suppressing the resulting degree-2 vertex, there exists an element $\ell$ in $X$ that is a leaf in $\cT'$ and not a leaf in $\cT''$. Let $X'$ and $X''$ be the leaf set of $\cT'$ and $\cT''$, respectively. Noting that $v_\cT$ is an $X$-vector that contains each element in $X$ exactly once, it follows from construction of the mapping that $v_{\cT'}$ is an $X'$-vector that contains each element in $X'$ exactly once and that $v_{\cT''}$ is an $X''$-vector that contains each element in $X''$ exactly once. Hence $v_{\cT'}\ne v_{\cT''}$. Moreover, since no element in $\cC(\cT)\setminus\{\emptyset\}$ is mapped to $v_\emptyset$, the mapping is one-to-one. This completes the proof of the lemma. \qed
\end{proof}

\begin{lemma}\label{lem:|C(T)|}
Let $\cT$ be a rooted binary phylogenetic $X$-tree. Then $$|\cC(\cT)|\leq (|X|+1)^{4c_\cT-2}.$$
\end{lemma}

\begin{proof}
Let 
\begin{eqnarray}
v_{\cT}=&(&[\xi_1](x_{1}^1,x_{1}^2,\ldots,x_{1}^{q_1},\epsilon),\nonumber\\
&&[\xi_2](x_{2}^1,x_{2}^2,\ldots,x_{2}^{q_2},\epsilon),\nonumber\\
&&\vdots \nonumber\\
&&[\xi_n](x_n^1,x_n^2,\ldots,x_n^{q_n},\epsilon)\hspace{0.3cm})\nonumber
\end{eqnarray}
be the vector representation of $\cT$ relative to an index tree of $\cT$, where $n=2c_\cT-1$. We first derive an upper bound on the number of $X$-vectors that satisfy the suffix-property relative to $v_\cT$. For each $i\in\{1,2,\ldots,n\}$, consider the vector component $[\xi_i](x_{i}^1,x_{i}^2,\ldots,x_{i}^{q_i},\epsilon)$. Then each $X$-vector that satisfies the suffix-property relative to $\cT$ has an $i$th vector component, say $[\psi_i](y_{i}^1,y_{i}^2,\ldots,y_{i}^{r_i},\epsilon)$, such that $\psi_i\in X\cup\{\epsilon\}$ and $(y_{i}^1,y_{i}^2,\ldots,y_{i}^{r_i},\epsilon)$ is a suffix of $(x_{i}^1,x_{i}^2,\ldots,x_{i}^{q_i},\epsilon)$. Since there are at most $|X|+1$ such suffixes, it follows that there are at most $(|X|+1)^2$ variations of $[\psi_i](y_{i}^1,y_{i}^2,\ldots,y_{i}^{r_i},\epsilon)$. Hence, there are at most 
\[
((|X| + 1)^2)^n = (|X| + 1)^{4c_\cT-2}.
\]
$X$-vectors that satisfy the suffix-property relative to $v_\cT$. By Lemma~\ref{l:map}, each tree in $\cC(\cT)$ can be mapped to one such vector and, as the map is one-to-one, it follows that $\cC(\cT)$ contains at most $(|X| + 1)^{4c_\cT-2}$ trees. \qed
\end{proof}

For a rooted binary phylogenetic $X$-tree $\cT$, the next lemma constructs an automaton that recognizes whether or not a word that contains each element in $X$ precisely once is a cherry-picking sequence for $\cT$. 


\begin{lemma}\label{lem:automaton}
Let $\cT$ be a rooted binary phylogenetic $X$-tree. There is a  deterministic finite automaton~$\cA_\cT$ with $O(|X|^{4c_\cT-2})$ states that recognizes the language 
\[
\cL_{X}(\cT) = \{x_1x_2\ldots x_{|X|} \mid (x_1,\,x_2,\,\ldots,\,x_{|X|}) \text{ is a cherry-picking sequence for } \cT \}.
\]
Moreover, the automaton~$\cA_\cT$ can be constructed in time $f(|X|, c_\cT) \in |X|^{O(c_\cT)}$.
\end{lemma}

\begin{proof}
Throughout this proof, we denote the tree without a vertex by $\emptyset$. Let $M$ and $M'$ be two sets. Setting $M = \{\cT\}$ and $M' = \emptyset$, we construct $\cA_\cT$  as follows.
\begin{enumerate}[(1)]
\item Create the states $q_\cT$, $q_\emptyset$, and $q_e$. For each $a \in X$, set $\delta(q_e, a) = \delta(q_\emptyset, a) = q_e$. 
\item For each $\cT \in M$ and each $a \in X$ do the following.
\begin{enumerate}[(a)]
\item If $a$ is a leaf of a cherry in $\cT$ or $a$ is the only vertex of $\cT$, then
\begin{enumerate}[(i)]
\item create the state $q_{\cT[-a]}$ if $\cT[-a]$ is not isomorphic to a tree in $M'$, 
\item set $M' = M' \cup \cT[-a]$, and
\item set $\delta(q_{\cT}, a) = q_{\cT[-a]}$. 
\end{enumerate}
\item Otherwise, set $\delta(q_{\cT}, a) = q_e$.
\end{enumerate}
\item Set  $M = M'$ and, subsequently, set $M' = \emptyset$. If $M \neq \{\emptyset\}$, continue with $(2)$.
\end{enumerate}
We set the initial state of $\cA_\cT$ to be $q_\cT$ and the final state to be $q_\emptyset$. To illustrate, the construction of $\cA_\cT$ is shown in Figure~\ref{fig:automaton} for a phylogenetic tree on four leaves.

\begin{figure}
\centering
\includegraphics[width=.78\textwidth]{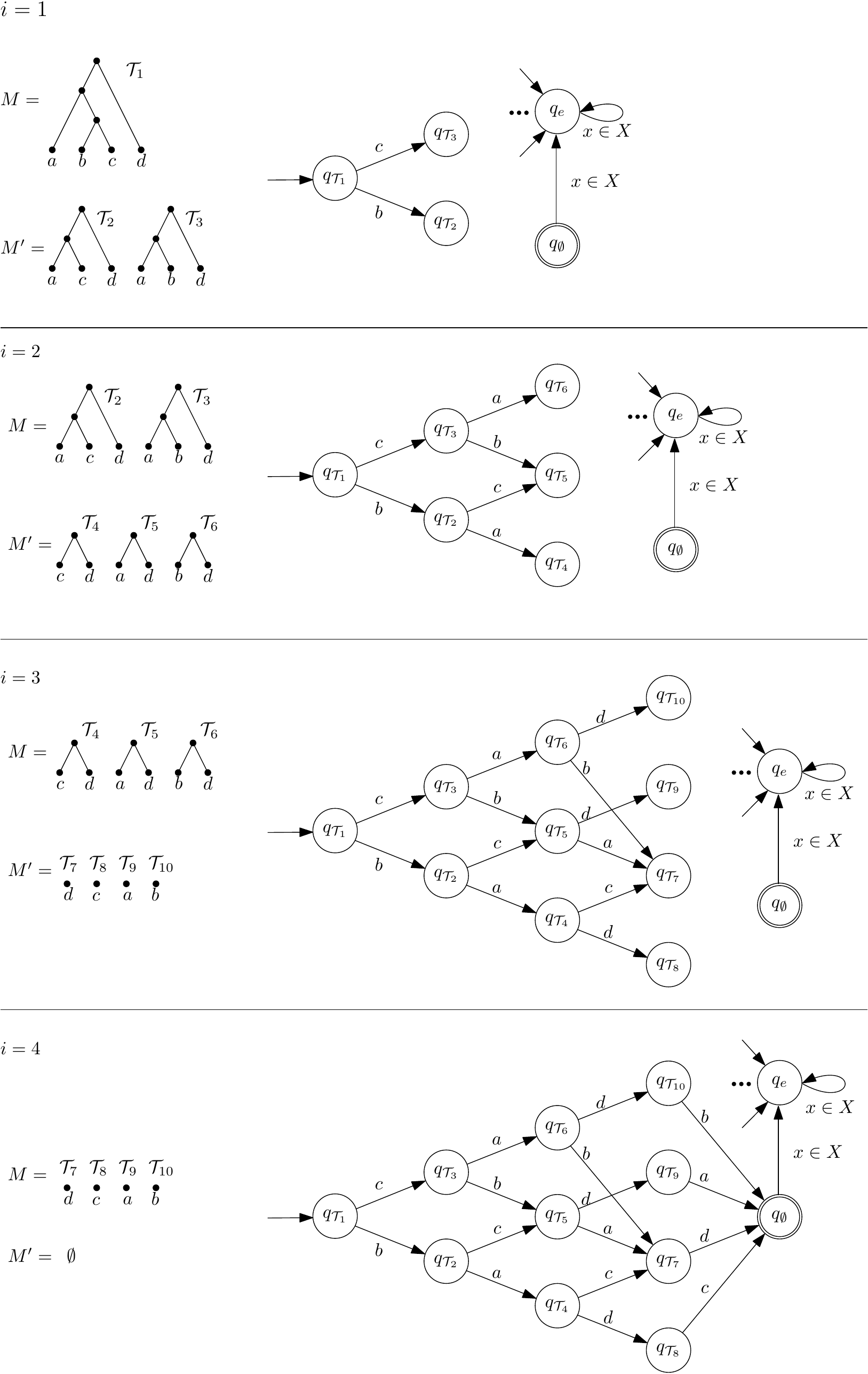}
\caption{Construction of an automaton that recognizes the language $\cL_X(\cT_1)$ as described in the statement of Lemma~\ref{lem:automaton} and with $\cT_1$ shown in the top left of this figure. Each vertex (resp. edge) represents a state (resp. transition). The vertex $q_{\cT_1}$ indicates the initial state whereas the final state is $q_\emptyset$ as indicated by a double circle. To increase readability, most transitions to $q_e$ are omitted.
In row $i$, the figure shows $M$, $M'$, and the automaton after the $i$th execution of the for-loop as  described in Step (2) in the proof of Lemma~\ref{lem:automaton}.}
\label{fig:automaton}
\end{figure}

By construction, we have $\cA_\cT=(\cC(\cT),X,\delta,q_\cT,\{q_\emptyset\})$. As each cherry-picked tree in $\cC(\cT)$ is mapped to a unique state, it follows from  Lemma~\ref{lem:|C(T)|} that the number of states of $\cA_\cT$ is $O(|X|^{4c_\cT-2})$.
Moreover, for each $a\in X$ and each pair of two distinct states $\cT',\cT'' \in \cC(\cT)$, there is a transition $\delta(q_{\cT'}, a) = q_{\cT''}$ if and only if $\cT'[-a] = \cT''$ and $a$ is either a leaf of a cherry in $\cT'$ or $\cT'$ consists of the single vertex~$a$. The state $q_e$ collects all inputs that do not correspond to the continuation of a cherry-picking sequence. More precisely, there is a transition $\delta(q_{\cT'}, a) = q_{e}$  if and only if $a$ is not a leaf of a cherry in $\cT'$ and $\cT'$ does not consist of the single vertex $a$.  It now follows that there is a one-to-one correspondence between the directed paths from $q_\cT$ to $q_\emptyset$ in $\cA_\cT$ and the cherry-picking sequences of $\cT$ and, hence, $\cA_\cT$ recognizes $\cL_{X}(\cT)$.  

The time taken to construct $\cA_\cT$ is dominated by the number of iterations of the for-loop in Step $(2)$. Since $|M| < |\cC(\cT)|$ and $|\cC(\cT)| \in O(|X|^{4c_\cT - 2})$, the number of iterations in Step (2) is $O(|\cC(\cT)|\cdot|X|) \subseteq O(|X|^{4c_\cT - 1})$. 
Moreover, since Step (2) is executed $|X|$ times, each operation of the for-loop is executed $O(|X|^{4c_\cT})$ times in total. While the complexity of these operations depend on the implementation and data structure, they can clearly be implemented such that $\cA_\cT$ can be constructed in time $|X|^{O(c_\cT)}$. This establishes the lemma. \qed
\end{proof}

Generalizing the language that is described in the statement of Lemma~\ref{lem:automaton}, the next straightforward observation describes a language for the decision problem {\sc CPS-Existence}.

\begin{observation}\label{obs:CPS_language_problem}
Let $\cP=\{\cT_1,\cT_2,\ldots,\cT_m\}$ be a collection of rooted binary phylogenetic $X$-trees. Then, solving {\sc CPS-Existence} for $\cP$ is equivalent to deciding if
\[
\bigcap_{1\leq i\leq m} \cL_X(\cT_i) \neq \emptyset.
\]
\end{observation}
We are now in a position to establish Theorem~\ref{t:cherry}.

\noindent{\it Proof of Theorem~\ref{t:cherry}.}
By Observation~\ref{obs:CPS_language_problem}, it follows that there is a cherry-picking sequence for $\cP$ if and only if
$
\bigcap_{\cT_i \in \cP} \cL_X(\cT_i) \neq \emptyset, 
$
where 
\[
\cL_{X}(\cT) = \{x_1x_2\ldots x_{|X|} \mid (x_1,\,x_2,\,\ldots,\,x_{|X|}) \text{ is a cherry-picking sequence for } \cT \}.
\]

For each $\cT_i \in \cP$ with $1 \leq i \leq m$, we follow the notation and construction that is described in the proof of Lemma~\ref{lem:automaton} to obtain an automaton~$\cA_{\cT_i}$ with $O(|X|^{4c_{\cT_i}-2})$ states that recognizes the language~$\cL_X(\cT_i)$. To solve the question whether or not the intersection of these $m$ languages is empty, we use the well-known construction of a product automaton \cite{kozen97} as follows.


For each $\cT_i \in \cP$, let $Q_{\cT_i}$ be set of states, and let $\delta_{\cT_i}$ be the transition relation of $\cA_{\cT_i}$. We construct a new automaton $\cA_{\cP}$, where the set of states $Q_\cP$ is the cartesian product $Q_{\cT_1} \times Q_{\cT_2}\times \ldots \times Q_{\cT_{m}}$. Furthermore, the alphabet of $\cA_{\cP}$ is $X$ and the transition relation $\delta_\cP \colon Q_\cP \times X \rightarrow Q_\cP$ is defined as 
\[
\delta_\cP ((q_1,\ldots,q_m), a) = (\delta_{\cT_1}(q_1, a),\, \ldots,\, \delta_{\cT_m}(q_m, a)).
\]
Lastly, the initial (resp. final) state of $\cA_\cP$ is $(q_1,\ldots,q_m)$ where, for all $i\in\{1,2,\ldots,m\}$, $q_i$ is the initial (resp. final) state of $\cA_{\cT_i}$. Intuitively, $\cA_\cP$ simulates the parallel execution of the automata $\cA_{\cT_1},\cA_{\cT_2},\ldots,\cA_{\cT_m}$. By construction, an input sequence is accepted by $\cA_\cP$ if and only if it is accepted by each automaton $\cA_{\cT_i}$. It now follows that there is a cherry-picking sequence for $\cP$ if and only if the final state of $\cA_\cP$ can be reached from the initial state of $\cA_\cP$ and, hence, $\cL(\cA_\cP) \neq \emptyset$. 


It remains to show that the computational complexity is as claimed in the statement of the theorem. Viewing $\cA_\cP$ as a directed graph, each directed path from the initial to the final state of $\cA_\cP$ has length $|X|$. We can therefore decide whether the final state from $\cA_\cP$ is reachable by using breadth-first search \cite{cormen09} in time $O(|Q_\cP| + |X| \cdot |Q_\cP|)$, where $|X| \cdot |Q_\cP|$ is the number of transitions in $\cA_\cP$. By construction and Lemma~\ref{lem:automaton}, it follows that $\cA_\cP$ has
\[
O\left(\prod_{i=1}^m |X|^{4c_{\cT_i}-2}\right) \subseteq O(|X|^{m(4c-2)})
\]  
states, i.e. $|Q_\cP| \in O(|X|^{m(4c-2)})$. Hence, we can decide in time $O(|X|^{m(4c-2)+1})$ whether $\cL(\cA_\cP) \neq \emptyset$. By Lemma~\ref{lem:automaton}, it takes time $f_i(|X|, c_{\cT_i}) \in |X|^{O(c_{\cT_i})}$ to construct each automaton $\cA_{\cT_i}$ and, thus, it follows that deciding if there is a cherry-picking sequence for $\cP$ can be done in time
\[
O\left (|X|^{m(4c-2)+1} + \sum_{i=1}^m f_i(|X|, c_{\cT_i})\right ),
\]
which is polynomial in $|X|$ if $c$ and $m$ are constant.
\qed\\

\section{Concluding remarks}\label{sec:con}
In this paper, we have shown that {\sc CPS-Existence}, a problem of relevance to the construction of phylogenetic networks from a set of phylogenetic trees, is NP-complete for all sets $\cP$ of rooted binary phylogenetic trees with $|\cP|\geq 8$. This result partially answers a question posed by Humphries et al.~\cite{humphries13}. They asked if {\sc CPS-Existence} is computationally hard for $|\cP|=2$. To establish our result, we first showed that {\sc $4$-Disjoint-Intermezzo}, which is a variant of the {\sc Intermezzo} problem that is new to this paper,  is NP-complete. Subsequently, we established a reduction from an instance $I$ of {\sc $4$-Disjoint-Intermezzo} to an instance $I'$ of {\sc CPS-Existence} with $|\cP|=8$. Since each of the four collections of pairs and triples in $I$ reduces to two trees in $I'$, a possible approach to obtain a stronger hardness result for {\sc CPS-Existence} with $|\cP|<8$ is to show that  {\sc $N$-Disjoint-Intermezzo} is NP-complete for $N<4$. However, it seems likely that such a result can only be achieved by following a strategy that is different from the one that we used in this paper. In particular, there is no obvious reduction from {\sc 2P2N-3-SAT} to {\sc $3$-Disjoint-Intermezzo}. Moreover, {\sc $1$-Disjoint-Intermezzo} is solvable in polynomial time since all pairs and triples are pairwise disjoint and, so, it cannot be used for a reduction even if {\sc CPS-Existence} turns out to be NP-complete for $|\cP|=2$.

In the second part of the paper, we have translated {\sc CPS-Existence} into an equivalent problem on languages and used automata theory to show that {\sc CPS-Existence} can be solved in polynomial time if the number of trees in $\cP$ and the number of cherries in each such tree are bounded by a constant. There are currently only a small number of other problems in phylogenetics that have been solved with the help of automata theory (e.g.~\cite{hall10,westesson12}) and it is to be hoped that the results presented in this paper will stimulate further research to explore  connections between  combinatorial problems in phylogenetics and automata theory. 

{\bf Acknowledgements} 

We thank Britta Dorn for insightful comments on a draft version of this paper. The second author was supported by the New Zealand Marsden Fund.

\end{document}